\documentclass[journal,twoside,web]{ieeecolor}
\usepackage{lcsys}
\usepackage{cite}
\usepackage{amsmath,amssymb,amsfonts}
\usepackage{algorithmic}
\usepackage{graphicx}
\usepackage{textcomp}

\usepackage{array}
\usepackage{multirow}

\newcolumntype{C}{>{\centering\arraybackslash}p{0.26cm}}
\newcolumntype{D}{>{\centering\arraybackslash}p{0.42cm}}

\newcommand{\Q}{\mathcal Q}
\newcommand{\G}{\mathcal G}

\newcommand{\K}{\mathcal K}
\newcommand{\R}{\mathbb R}

\usepackage{amsthm}
\newtheorem{thm}{Theorem}
\newtheorem{prop}[thm]{Proposition}

\theoremstyle{remark}
\newtheorem{remark}{Remark}

\begin{document}
\title{Learning over All Stabilizing Nonlinear Controllers for a Partially-Observed Linear System}
\author{Ruigang Wang, Nicholas H. Barbara, Max Revay, and Ian R. Manchester
\thanks{This work was supported in part by the Australian Research Council. }
\thanks{The authors are with the Sydney Institute for Robotics and Intelligent Systems, Australian Centre for Field Robotics, and also with the School of Aerospace, Mechanical and Mechatronic Engineering, The University of Sydney, Sydney, NSW 2006, Australia (e-mail: ian.manchester@sydney.edu.au).}}

\maketitle
\thispagestyle{empty} 

\begin{abstract}
This paper proposes a nonlinear policy architecture for control of partially-observed linear dynamical systems providing built-in closed-loop stability guarantees. The policy is based on a nonlinear version of the Youla parameterization, and augments a known stabilizing linear controller with a nonlinear operator from a recently developed class of dynamic neural network models called the recurrent equilibrium network (REN). {We prove that RENs are universal approximators of contracting and Lipschitz nonlinear systems, and subsequently show that the the proposed Youla-REN architecture is a universal approximator of  stabilizing nonlinear controllers}. The REN architecture simplifies learning since unconstrained optimization can be applied, and we consider both a model-based case where exact gradients are available and reinforcement learning using random search with zeroth-order oracles. In simulation examples our method converges faster to better controllers and is more scalable than existing methods, while guaranteeing stability during learning transients.
\end{abstract}

\begin{IEEEkeywords}
Contraction, learning based control, Youla parameterization, nonlinear output feedback
\end{IEEEkeywords}

\section{Introduction}
\label{sec:introduction}
\IEEEPARstart{D}{eep} neural networks and reinforcement learning (RL) hold tremendous potential for continuous control applications and impressive results
have already been demonstrated  on robot locomotion, manipulation, and other benchmark control tasks  \cite{lillicrap2016continuous}. Nevertheless, deep RL based controllers are not yet widely used in engineering practice, as they can exhibit unexpected or brittle behaviour. This has led to a rapid growth in literature studying computational verification of neural networks and reinforcement learning with stability guarantees.

Research in stable deep RL has largely assumed full state information is available \cite{berkenkamp2017safe,chang2019neural,khader2021learning}.
In practice, most controllers only have access to partial state observations. In this setting, the optimal controller is generally a dynamic function of previous observations. To the authors' knowledge, the only existing methods for this setting are
\cite{knight2011stable,gu2021recurrent}. These methods propose nonconvex sets of stabilizing recurrent neural network (RNN) controllers that require nontrivial projection methods at each training step. 

\subsection{Background}

\paragraph{The Youla Parameterization}
Developed in the late 70's, the Youla parameterization (or the Youla-Kucera parameterization) provides a parameterization of all stabilizing controllers for a given LTI system via a so-called \textit{$\mathcal Q$ parameter} \cite{youla1976modernI}, not to be confused with the state-action value function in RL, which is also usually denoted by $Q$. The key feature of the Youla parameterization is that all stabilizing controllers can be represented via a stable system $\mathcal Q$, and the closed-loop response is linear in $\mathcal Q$. It plays a central role in linear robust control theory \cite{zhouRobustOptimalControl1996}, controller optimization \cite{boyd1991linear}, and decentralized control \cite{rotkowitz2005characterization}. It has been used to guarantee stability in the linear reinforcement learning setting \cite{roberts2011feedback}, to give the tightest-known regret bounds for online learning of linear controllers \cite{simchowitz2020improper}, and has been used in many applications, see e.g. \cite{mahtout2020advances} for a recent review.

To date, these applications have employed a linear $\mathcal Q$ parameter despite several works extending the theory to a nonlinear setting, see e.g. \cite{fujimotoCharacterizationAllNonlinear2000}. The main barriers to the widespread application of nonlinear Youla parameterizations have been (i) non-constructive parameterizations expressed in terms of coprime factors or kernel representations which are difficult to compute, and (ii) the lack of flexible parameterizations of nonlinear $\mathcal Q$ parameters.

\paragraph{Learning Stable Dynamical Models}
The construction of a nonlinear $\mathcal Q$ parameter requires a parameterization of nonlinear systems with stability guarantees. A simple approach is to bound the maximum singular value of an RNN's weight matrix leading to simple, yet conservative stability criteria \cite{jaeger2002adaptive,miller2018stable}. There have been a number of improvements that reduce conservatism by using non-Euclidean metrics \cite{revay2020contracting}, incremental integral quadratic constraints (IQCs) \cite{revay2020convex}, and local IQCs \cite{yin2021stability}. A more general parameterization allowing state-dependent contraction metrics was given in \cite{tobenkin2017convex}. A fundamental problem with these approaches is that the stability constraints take the form of (possibly nonconvex) matrix inequalities which require complex projection, barrier, or penalty methods to enforce. A significant development addressing these issues was the \textit{direct parameterization} developed in \cite{revay2021recurrent}, providing an \textit{unconstrained} model parameterization with stability and robustness guarantees.

\paragraph{Reinforcement Learning and Random Search}
RL is the problem of learning an approximate optimal controller given only observations or simulations of the process, not explicit models. There has been an enormous number of new algorithms for RL proposed in recent years, however the vast majority  assume that the full state is available to the control policy during learning and execution \cite{sutton2018reinforcement}. Many methods work by fitting approximate value functions over the state space, while standard policy gradient algorithms rely on a Markovian assumption \cite{sutton2018reinforcement}.  It has recently been demonstrated that in many benchmark control problems, simple random search over policy space is highly competitive \cite{Mania2018advances}. The algorithm is a  modified version of the classic stochastic approximation method of \cite{robbins1951stochastic}.  Random search takes random jumps in parameter space and  approximates directional derivatives from cost evaluations. It  makes no assumptions on the controller structure or state information, and is therefore well-suited to training dynamic controllers on partially observed systems.

\subsection{Contributions}
This paper proposes a method to learn over all stabilizing nonlinear dynamic controllers for a partially-observed linear system  via unconstrained optimization. We build on our recent work \cite{revay2021recurrent, wang2021yoularen,revay2021recurrent-full} proposing the recurrent equilibrium network (REN) model structure, and incorporating it into control systems via the Youla parameterization. The paper \cite{gu2021recurrent} considers a similar problem statement, but requires a solving a semidefinite program (SDP) at each gradient step to project onto a convex approximation of a non-convex set. This paper provides the following contributions relative to prior work:
\begin{enumerate}
    \item We give necessary and sufficient conditions for closed-loop stability, improving on the sufficient conditions in \cite{revay2021recurrent-full, wang2021yoularen, gu2021recurrent}.
    \item {Whereas \cite{revay2021recurrent, revay2021recurrent-full} give a construction of contracting and Lipschitz RENs, in this paper we prove that the REN architecture is a universal approximator of contracting and Lipschitz  nonlinear systems.}
    \item We learn all parameters of the REN model, not just the output layer as in \cite{revay2021recurrent-full}.
    \item We consider the output-feedback case, not the state-feedback case as in \cite{wang2021yoularen, revay2021recurrent-full}.
    \item Our method uses unconstrained optimization, while \cite{gu2021recurrent} requires solving an SDP at each gradient step.
    \item We show that our method is naturally suited to gradient-free optimization  via random search, since it can guarantee stability during the search procedure.
\end{enumerate}

\section{Problem Setup}

We consider a discrete-time linear time-invariant system:
\begin{align}
  x_+ = Ax+Bu+d_x, \quad
  y = Cx+d_y
\end{align}
with internal state $x$, controlled inputs $u$, measured outputs $y$, and disturbances/noise on the state and measurements $d_x, d_y$, respectively. For brevity $x=x_t, x_+=x_{t+1}$, and similarly for other variables. Define 
\[d:=\begin{bmatrix}d_x\\d_y\end{bmatrix}, \quad z:=\begin{bmatrix}x\\u\end{bmatrix}.\]
We assume that the system is stabilizable and detectable, and that the disturbance signal $d$ and initial state $x_0$ are random variables for which a sampler is available for training.

We consider the problem of finding a controller  $u=\K_{\theta}(y)$ with $\theta$ as the learnable parameter, which may be nonlinear and/or dynamic (i.e. have internal states), such that:
\begin{enumerate}
    \item The closed-loop system is contracting, i.e. initial conditions are forgotten exponentially \cite{lohmiller1998contraction}. 
    
    \item The closed-loop response to disturbances is Lipschitz (i.e. has bounded incremental $\ell^2$ gain):
    \[\|z^a-z^b\|\le \gamma \|d^a-d^b\|\]
    for some $\gamma>0$ where $\|\cdot\|$ is the signal $\ell^2$ norm $\|z\| = \sqrt{\sum_{t=0}^\infty |z_t|^2}$, $d^a, d^b$ are two realizations of the disturbance, and $z^a, z^b$ are the corresponding realizations of $z$.
    \item A cost function of the following form is minimized (at least approximately and locally)
\begin{equation} \label{eqn:expected_cost}
    J_\theta = E\left[\sum_{t=0}^{T-1}g(x_t, u_t) + g_T(x_T)\right],
\end{equation}
where $E[\cdot]$ is expectation over  $x_0$ and $d$. 
\end{enumerate}
Note that the first two requirements are \textit{hard} requirements that must be satisfied, whereas the third is \textit{soft} in the sense that we do not expect a global minimum to necessarily be found, since we do not make any assumptions about convexity of $g, g_T$, nor the ability to precisely compute the expectation. In this work, the expectation is  approximated by averaging over a finite batch of sampled $x_0$ and $d$, denoted by  $\hat J_\theta$.

We will consider two scenarios that differ in terms of the information available to the learning algorithm:
\begin{enumerate}
    \item A first-order oracle is available, i.e. where $\hat J_\theta$ and $ \nabla_\theta \hat{J}_\theta$ can be evaluated at any $\theta$. This is implementable when controllers are optimized using a differentiable simulation model, e.g. \cite{howell2022dojo}, hence we refer to this as the ``model-based case''.  
    \item A zeroth-order oracle is available, i.e. only $\hat{J}_\theta$ can be evaluated at each $\theta$. This is the setting usually assumed in reinforcement learning, since in principle it is implementable in experiments. In this case, we will approximate gradients using the method of \cite{Mania2018advances}.
\end{enumerate}

\section{The Youla-REN Controller Architecture}

We will construct a parameterization of all possible nonlinear controllers such that the closed-loop system is contracting and the mapping $d\mapsto z$ is Lipschitz in terms of a ``parameter'' which is itself a contracting and Lipschitz nonlinear system. We then use the REN model introduced in \cite{revay2021recurrent, revay2021recurrent-full} as this parameter.

Since the system is stabilizable and detectable, one can compute gain matrices $K,L$ such that $(A-BK)$ and $(A-LC)$ are stable  (e.g. LQG \cite{zhouRobustOptimalControl1996}). We assume some such gains are known and define the ``base''  linear controller 
\begin{align}
    \hat x_+ = A\hat x+Bu_K+L\tilde y, \quad u_K  = -K\hat x,\label{eq:base}
\end{align}
where $\tilde y = y-C\hat x$. The proposed control architecture is an extension of the classical Youla parameterization for linear systems, and works by augmenting the base controller:
\begin{align}
    \hat x_+ = A\hat x+Bu+L\tilde y, \quad u  = -K\hat x+\Q(\tilde y),\label{eq:qparam}
\end{align}
where $\Q$ is an arbitrary contracting and Lipschitz nonlinear system. Since $\tilde y$ represents the difference between expected and observed outputs, the Youla parameterization could be interpreted as prescribing a \textit{stable response to surprises}.

\begin{remark}
The policy \eqref{eq:qparam} can be extended to incorporate exogenous signals $r$ (e.g. reference signals, feedforward commands, disturbance previews, parameter estimates) by using a controller the form $u=-K\hat x+\Q(\tilde y, r)$, where $\Q$ is contracting and Lipschitz in both inputs. All the following theoretical results apply unchanged in this case, but in this paper we restrict to the pure-feedback case for brevity of notation.
\end{remark}
 
\subsection{Theoretical Results}
We will show that our proposed parameterization is in a sense universal, i.e., any stabilizing dynamic output-feedback controller can be parameterized via $\Q$. 

Consider an arbitrary feedback controller $u= \K(y)$ admitting a state-space realisation
\begin{align}
  \zeta_+=f(\zeta,y), \quad u = g(\zeta,y) \label{eq:nlc}
\end{align}
where $\zeta$ is the state and $f,g$ are locally Lipschitz, leading to the closed-loop dynamics:
\begin{equation}\label{eq:cl}
  \begin{split}
      x_+ &=Ax+d_x+Bg(\zeta,Cx+d_y), \\
  \zeta_+ &= f(\zeta, Cx+d_y).
  \end{split}
\end{equation}
We give the following version of the well-known universality property of the Youla parameterization.
\begin{prop}\label{prop:1}
	Consider the following control architecture \eqref{eq:qparam}, 
parameterized by $\Q$.\begin{enumerate}
  \item For any contracting and Lipschitz $\Q$, the closed-loop system with the controller \eqref{eq:qparam} is contracting and Lipschitz.
  \item Any controller of the form \eqref{eq:nlc} that achieves contracting and Lipschitz closed-loop can be written in the form \eqref{eq:qparam} with contracting and Lipschitz $\Q$.
\end{enumerate}
Moreover, the closed-loop response with this controller structure is
\begin{equation}
 z= \mathcal{T}_0d+\mathcal{T}_1\Q(\mathcal{T}_2d)\label{eq:zTQ}
\end{equation}
where $\mathcal{T}_0,\mathcal{T}_1,\mathcal{T}_2$ are stable linear systems.
\end{prop}

\noindent\textbf{Proof:} 
We first construct $\mathcal{T}_0,\mathcal{T}_1,\mathcal{T}_2$, as per \cite{zhouRobustOptimalControl1996}. Defining  $\tilde x = x-\hat x$, the closed-loop dynamics under the base controller \eqref{eq:base} can be written as
\begin{equation}\label{eq:lqg_rewrite}
    \begin{split}
        \tilde x_+ &= (A-LC)\tilde x+d_x-Ld_y,\\
x_+  &=(A-BK)x+d_x+BK\tilde x,\\
  \tilde y &= C\tilde x, \quad
  z = \begin{bmatrix}
      x^T,& (K(\tilde{x}-x))^T
  \end{bmatrix}^T.
    \end{split}
\end{equation}
The first, second, and fourth equations define the stable closed-loop response with the base controller, $d\mapsto z$, which we denote $\mathcal{T}_0$. The first and third define a stable linear system $d\mapsto \tilde y$, which we denote $\mathcal{T}_2$. To construct  $\mathcal{T}_1$, we introduce a virtual control input: 
$
\tilde u :=u-u_K = u+K\hat x=u+K(x-\tilde x) 
$
Now, the plant dynamics under the Youla controller \eqref{eq:qparam} can be rewritten as a linear system:
\begin{align}
  x_+ 
  & = (A-BK)x+d_x+BK\tilde x + B\tilde u. \label{eq:xv}
\end{align}
So by superposition, we have
$
 z= \mathcal{T}_0{d}+\mathcal{T}_1\tilde u
$
where $\mathcal{T}_1$ is the stable system
\begin{align}
 \xi_+ &= (A-BK)\xi+B\tilde u,\quad
 \eta= [\xi^T,\  (\tilde u-K\xi)^T]^T.\notag 
\end{align}
Hence if $\Q$ maps $\tilde y$ to $\tilde u$, we have the closed loop \eqref{eq:zTQ}.

\textit{Proof of Claim 1.}
This follows from the stability of $\mathcal{T}_0,\mathcal{T}_1,\mathcal{T}_2$ and the contraction and Lipschitz condition on $\Q$, and the composition properties of contracting systems and Lipschitz mappings.

\textit{Proof of Claim 2.} Assuming a  controller $u=\K(y)$ exists in the form \eqref{eq:nlc} we can equivalently augment it with the state estimator and rewrite it in terms of $\tilde y$ and $\tilde u$ as follows:
\begin{align}
\hat x_+ &=A\hat x+L\tilde y+Bg(\phi,  C\hat x+ \tilde y)\label{eq:QK1}\\
  \phi_+&= f(\phi, C\hat x+ \tilde y)\label{eq:QK2} \\
  \tilde u &=  K\hat x+u=K\hat x+g(\phi,C\hat x+\tilde y)  \label{eq:QK3}
\end{align}
defining the closed-loop mapping $\Q_\K:\tilde y\mapsto \tilde u$. 

Now, taking \eqref{eq:cl}, relabelling $x$ as $\hat x$, and taking $d_x=L\tilde y$ and $d_y=\tilde y$ as particular inputs, gives \eqref{eq:QK1}, \eqref{eq:QK2}. Hence if the closed-loop with $\K$ is contracting then so is $\Q_\K$. Moreover, the fact that the mapping from $d\to z$ is Lipschitz  with  \eqref{eq:cl} implies that the mapping $\tilde y \to v$ is too, since by \eqref{eq:QK3} $v= [K, I]z$ under this relabelling so $|\tilde u|\le \alpha |z|$ for some $\alpha$. Hence the closed-loop system is contracting and Lipschitz if and only if $\Q_\K$ is.\hfill$\Box$

\subsection{$\Q$ parameterization via RENs}
In this paper, we parameterize $\Q$ via \textit{recurrent equilibrium networks} (RENs), a model architecture introduced in \cite{revay2021recurrent}. We will show that a REN is an universal approximator for the contracting and Lipschitz Youla parameter. 

We first review the REN model $\tilde \Q:\tilde y \mapsto \tilde u $, which is itself a feedback interconnection of a linear system and nonlinear ``activation functions'' $\sigma$:
\begin{equation}\label{eq:ren}
   \tilde \Q
    \begin{cases}
        \begin{bmatrix}
            \chi_{t+1} \\ v_t \\ \tilde u_t
        \end{bmatrix}=
        \overset{W}{\overbrace{
		\left[
            \begin{array}{c|cc}
            A_\chi & B_1 & B_2 \\ \hline 
            C_{1} & D_{11} & D_{12} \\
            C_{2} & D_{21} & D_{22}
		\end{array} 
		\right]
        }}
        \begin{bmatrix}
            \chi_t \\ w_t \\ \tilde y_t
        \end{bmatrix}+
        \overset{b}{\overbrace{
            \begin{bmatrix}
                b_\chi \\ b_v \\ b_y
            \end{bmatrix}
        }},\vspace{10pt}\\
        w_t=\sigma(v_t):=
        \begin{bmatrix}
            \sigma(v_{t}^1) & \sigma(v_{t}^2) & \cdots & \sigma(v_{t}^q)
        \end{bmatrix}^\top
    \end{cases}
\end{equation} 
where $ \chi_t {\in \R^{n_\chi}}$ is the internal state, $ v_t,w_t\in \R^{n_v}$ are the input and output of the neuron layer. We assume that $ \sigma:\R\rightarrow\R$ is { a non-polynomial function with} slope restricted in $[0,1]$. When $D_{11}\ne 0$, the mappting $v\mapsto w$ defines an \textit{equilibrium network} a.k.a. \textit{implicit network}, although when $D_{11}$ is strictly lower-triangular this mapping is explicit. { Many feed-forward networks (e.g. multilayer perceptron, residual networks) can be represented as equilibrium networks \cite{revay2021recurrent-full}. }  

The learnable parameters in \eqref{eq:ren} are $(W, b)$, but the key feature of RENs relevant to this paper is that they admit a \textit{direct} parameterization, i.e. there is a smooth mapping $\theta\mapsto (W, b)$ from an unconstrained parameter $\theta \in \mathbb R^q$, such that for all $\theta$ the resulting REN is contracting and Lipschitz. A further benefit is that they are very flexible, including many previously-used models as special cases. The following proposition shows that {RENs are universal approximators for contracting and Lipschitz systems.}
\begin{prop}
    { For any contracting and Lipschitz $\Q:\tilde y\mapsto \tilde u$, and any $M, \epsilon>0$ there exists a sufficiently large $n_\chi$ and $n_v$ such that a REN $\tilde{\Q}$ exists with these dimensions and $\| \tilde{\Q}(\tilde y)-\Q(\tilde y)\|_\infty \leq \epsilon$ for all $\|\tilde y\|_\infty \leq M$.
    }\label{pr:universal}
\end{prop}

\begin{proof}
We first show that {$\Q$ has} fading memory in the sense of \cite{boyd1985fading}. { Assume that $\Q$ has the state-space representation $x_+=f(x,\tilde y),\, \tilde u=g(x,\tilde y)$} where $f,g$ are locally Lipschitz. Given an initial state $a$ and input sequence $\tilde y$, the corresponding state and output at time $t$ are denoted by $x_t^{a,\tilde y}$ and $\tilde u_t^{a,\tilde y}$, respectively. Since $\Q$ is contracting, there exists an incremental Lyapunov function $c_1|x_1-x_2|^2\leq V_1(x_1,x_2)\leq c_2|x_1-x_2|^2$ with $c_2\geq c_1>0$ such that for any initial state $a, b$ and input {$\|\tilde y^1\|_\infty \leq M$, we have}
$
    V_1\bigl(x_{t+1}^{a,\tilde y^1},x_{t+1}^{b,\tilde y^1}\bigr)\leq \alpha V_1\bigl(x_t^{a,\tilde y^1},x_t^{b,\tilde y^1}\bigr)
$
for some $\alpha \in (0,1)$. By applying the above inequality $N>\log(c_1/c_2)/\log\alpha$ steps from $t=0$, we have
\begin{equation}\label{eq:contraction}
    \left|x_{N}^{a,\tilde y^1}-x_{N}^{b,\tilde y^1}\right|\leq \beta \left|x_0^{a,\tilde y^1}-x_0^{b,\tilde y^1}\right|
\end{equation}
where $\beta = \alpha^{N/2}\sqrt{c_2/c_1}\in (0,1)$. Due to the Lipschitzness of $\Q$, there exists an incremental Lyapunov function $c_3|x_1-x_2|^2\leq V_2(x_1,x_2)\leq c_4|x_1-x_2|^2$ with $c_4\geq c_3>0$ such that for any initial state $b$ and inputs {$\|\tilde y^1\|_\infty,\|\tilde y^2\|_\infty\leq M$,}
\[
   \begin{split}
       { V_{2,t+1}-V_{2,t}}\leq \gamma^2|\tilde y_t^1-\tilde y_t^2|^2-\left|\tilde u_t^{b,\tilde y^1}-\tilde u_t^{b,\tilde y^2}\right|^2
   \end{split}
\]
for some $\gamma > 0$, {where $V_{2,t}=V_2\left(x_t^{b,\tilde y^1},x_t^{b,\tilde y^2}\right)$}.
Summation of the above inequality from $t=0$ to $t=N$ yields
\begin{equation}\label{eq:lipschitz}
        \left|x_{N}^{b,\tilde y^1}-x_{N}^{b,\tilde y^2}\right| \leq \sqrt{\frac{V_{2,N}}{c_3}} \leq
         \frac{\gamma}{\sqrt{c_3}}\sum_{t=0}^{N-1}|\tilde y_t^1-\tilde y_t^2|
\end{equation}
From \eqref{eq:contraction} and \eqref{eq:lipschitz} we have
$
    |x_N^{a,\tilde y^1}-x_N^{b,\tilde y^2}|\leq \beta |x_0^{a,\tilde y^1}-x_0^{b,\tilde y^1}|+\frac{\gamma}{\sqrt{c_3}}\sum_{t=0}^{N-1}|\tilde y_t^1-\tilde y_t^2|.
$
Shifting the initial time point to $t_0=-kN$ with $k\rightarrow\infty$ yields
\[
    \begin{split}
        \left|x_0^{a,\tilde y^1}-x_0^{b,\tilde y^2}\right|\leq 
        & \lim_{k\rightarrow\infty}\frac{\gamma}{\sqrt{c_3}}\sum_{j=0}^{k-1}\beta^{j}\sum_{t=-(j+1)N+1}^{-jN}|\tilde y_t^1-\tilde y_t^2| \\
        = &\lim_{k\rightarrow\infty}\frac{\gamma}{\sqrt{c_3}}\sum_{j=0}^{k-1}\beta_1^{j}\sum_{t=-(j+1)N+1}^{-jN}\beta_2^j|\tilde y_t^1-\tilde y_t^2| \\
        \leq & \frac{N\gamma}{(1-\beta_1)\sqrt{c_3}}\sup_{t\leq 0} \beta_2^{-\left\lfloor\frac{t}{N}\right\rfloor}|\tilde y_t^1-\tilde y_t^2|
    \end{split}
\]
where $\beta_1,\beta_2\in (0,1)$ and $\beta_1\beta_2=\beta$. Furthermore, we have
\begin{equation}\label{eq:fade-mem}
    \left|\tilde u_0^{a,\tilde y^1}-\tilde u_0^{b,\tilde y^2}\right|\leq \kappa \sup_{t\leq 0} \beta_2^{-\left\lfloor\frac{t}{N}\right\rfloor}|\tilde y_t^1-\tilde y_t^2|
\end{equation}
where $\kappa =L_2+\frac{L_1 N\gamma}{(1-\beta_1)\sqrt{c_3}}$ with $L_1,L_2$ as the local Lipschitz constants of $g$ with respect to $x$ and $\tilde y$, respectively. 

Eq.~\eqref{eq:fade-mem} implies that $ \G $ has the fading memory property, { and thus it} can be universally approximated by nonlinear moving averaged operators (NLMAs) \cite[Thm.~3]{boyd1985fading}. Since REN includes the NLMA structure { with equilibrium network as output mapping \cite[Sec. VI]{revay2021recurrent-full}, $\tilde{\Q}$ can approximate $\Q$ arbitrarily close as $n_\chi$ and $n_v$ increase. }
\end{proof}

\section{Simulation Experiments on Cart-Pole}\label{sec:experiment}

In this section, we compare empirical performance of three control structures for an open-loop unstable system (the linearized cart-pole system): 
\begin{itemize}
    \item Youla parameterisation \eqref{eq:qparam} with a learnable $\Q$-parameter;
    \item ``Feedback policy" wraps a learnable dynamic controller as an outer loop around the base controller \eqref{eq:base}, i.e. $u = -K\hat x+\K_\theta(y)$.
    \item Projection policy \eqref{eq:nlc} parameterized by an RNN \cite{gu2021recurrent}. 
\end{itemize}
For the learnable parameter in the first two structures we considered both dynamical models (REN and LSTM) and static mappings (LBEN and DNN). A Lipschitz-bounded equilibrium network (LBEN) is a memoryless version of the REN \cite{revay2021recurrent-full}, i.e., $n_\chi= 0$ in \eqref{eq:ren}. Each controller is denoted by \textit{structure-model}, e.g., the proposed approach is called Youla-REN. Note that neither the Youla-LSTM nor any of the ``Feedback'' policies provide closed-loop stability guarantees. All models were coded in Julia except the Project-RNN, for which we used Python sub-routines from \cite{gu2021recurrent} to perform model projection. 

\subsection{Problem Setup}

We consider a linearized cartpole system \cite{gu2021recurrent} with 
\[
\begin{split}
    A&=\begin{bmatrix}
		1 & \delta & 0 & 0 \\
		0 & 1 & -\frac{ m_p g\delta}{m_c} & 0\\
		0 & 0 & 1 & \delta\\
		0 & 0 & \frac{(m_c+m_p)g\delta}{l m_c} & 1
	\end{bmatrix},\;
B=\begin{bmatrix}
		0 \\ \frac{\delta}{m_c} \\ 0 \\ -\frac{\delta}{m_c}
	\end{bmatrix}, \\
C&=\begin{bmatrix}
	    1 & 0 & 0 & 0 \\ 0 & 0 & 1 & 0
	\end{bmatrix},
\end{split}	
\]
where $ m_p=0.2$, $m_c= 1.0$, $ l=0.5$, $g=9.81$ and $\delta=0.08$. We used an infinite-horizon LQG controller with random weight and covariance matrices as the base controller \eqref{eq:base}. We examined the following two control tasks. 

\paragraph{Task 1} Consider the classic LQG control problem:
\begin{equation}
	J_q=E\left[\|x_T\|_{Q_f}^2+\sum_{t=0}^{T-1}\bigl(\|x_t\|_Q^2 + \|u_t\|_R^2\bigr)\right]
\end{equation}
where $x_t,u_t$ are the state and input at time $t$, respectively. The weight matrices are given by $Q_f=Q=\mathrm{diag}(1,1,5,1)$ and $R=1$. The covariance matrices of $d_x$ and $d_y$ are chosen as $\Sigma_x=0.005I$ and $\Sigma_y=0.001I$, respectively. We used a trajectory length of $T=50$.

\paragraph{Task 2} We also considered the disturbance rejection problem with a convex but non-quadratic cost function
\begin{equation}
    J_c= E\left[ J_q+ \rho\sum_{t=0}^{T-1} \max(|u_t|-\overline{u}, 0)\right]
\end{equation}
where $ \overline{u}$ is the soft input bound and $\rho$ is the penalty coefficient. The optimal policy for this problem is nonlinear due to the input constraint. We chose a large $\rho=400$ such that $u_t$ could only exceed the bound $\overline{u}=2$ for short periods of time. This is relevant e.g. for electric motors which allow larger currents for short periods before reverting to their continuous current capacity. The input disturbances contain Gaussian noise and piecewise-constant perturbation with random duration and magnitude. The trajectory length chosen was $T=100$.

\subsection{Training details}

We chose models with similar numbers of parameters, corresponding to 10 states and 20 neurons for both REN and RNN, 10 cell units for LSTM, and 40 hidden neurons for LBEN and DNN, respectively. We used {\tt ReLU} activation functions for all models except LSTM. Each model was trained with the ADAM optimizer \cite{Kingma2015adam} over a range of learning rates from $0.001$ to $0.1$. Each experiment was repeated for 10 times with different random seeds, which affect model initialisation, noise, and disturbances. The learning rate was reduced by a factor of 10 after 85\% of the total training epochs. Gradients were clipped to an $\ell^2$ norm of 10.0 over all parameters. Depending on the learning scenarios, the policy gradients were computed in the following two ways.

\paragraph{Exact Gradients}
In the model-based case, we used the exact gradients, i.e., first-order oracle $\nabla_\theta \hat{J}_\theta$ with batch size of 40. Our chosen learning rates are provided in Table~\ref{tab:ars-stepsize}. We chose a learning rate of 0.1 for the Projection-RNN controller after finding its performance was insensitive to hyperparameter choices. 

\begin{table}[!bt]
    \centering
    \caption{Learning rates used for model training}
    \label{tab:ars-stepsize}
    \begin{tabular}{l|DDDD|DDDD}
         & \multicolumn{4}{c|}{LQG} & \multicolumn{4}{c}{Input constrained} \\
          \hline
         Exact grad. & DNN & LBEN & LSTM & REN & DNN & LBEN & LSTM & REN \\
         \hline
    Youla  & 0.01  & 0.01  & 0.1 & 0.01 & 0.01 & 0.01 & 0.1 & 0.01\\
    Feedback & 0.01  & 0.001  & 0.1 & 0.001 & 0.01 & 0.001 & 0.1 & 0.001\\
    \hline\hline
    Appr. grad. & DNN & LBEN & LSTM & REN & DNN & LBEN & LSTM & REN \\
         \hline
    Youla  & 0.01  & 0.01  & 0.04 & 0.01 & 0.01 & 0.02 & 0.04 & 0.02\\
    Feedback & 0.01  & 0.008 & 0.04 & 0.008 & 0.01 & 0.01 & 0.02 & 0.01 \\
    \end{tabular}
\end{table}

\paragraph{Approximate Gradients}

In the reinforcement learning case, only the zero-order oracle $\hat{J}_\theta$ is available. We approximated the policy gradients using a finite-difference approach called Augmented Random Search (ARS) \cite{Mania2018advances}. We made three minor changes to the original ARS algorithm to improve  convergence. Firstly, we used the ADAM optimiser and gradient clipping instead of free stochastic gradient descent. Secondly, we used the average cost over $b$ batches of initial conditions to approximate the expected cost. Thirdly, we used the same batch of initial conditions when evaluating the cost for perturbations of the model parameters for a fair comparison between policy rollouts. The estimated policy gradient with respect to the model parameters $\theta$ is then
\begin{equation}
	\nabla_\theta \hat{J}_\theta = \frac{1}{m} \sum_{i=1}^m \frac{\hat{J}_{\theta + \nu \delta_i}(\mathbf{x_{0}^i}) - \hat{J}_{\theta - \nu \delta_i}(\mathbf{x_{0}^i})}{2 \nu \sigma_R}\delta_i
\end{equation}
with $\hat{J}_\theta (\mathbf{x_0}) = \frac{1}{b} \sum_{j=1}^{b} J_\theta (x_{0}^j)$, 
where $x_{0}^j$ are initial state vectors, $\nu$ is a small perturbation, $\delta_i$ is a normally-distributed random vector with zero mean, and $\sigma_R$ is the standard deviation of the $2m$ average costs collected for each gradient approximation \cite{Mania2018advances}. We chose $(m, b, \nu) = (40, 10, 0.01)$. The chosen learning rates are provided in Table~\ref{tab:ars-stepsize}. It is not obvious how to effectively combine ARS with the iterative SDP projection step required by the Projection-RNN model in \cite{gu2021recurrent}. We therefore excluded it from our experiments.

\subsection{Results} \label{sec:Results}

\paragraph{Exact Gradients}
Fig.~\ref{fig:result} shows the training results of different controllers. Our Youla-REN performed the best overall, reaching the lowest final cost in the fewest training epochs. The Project-RNN controller seemed to fall into local minima after the first projection, making its cost more sensitive to model initialisation. The Youla-REN and Youla-LSTM controllers achieved similar final costs, but the Youla-REN converged faster and also guarantees stability during learning. Both Youla-LBEN and Youla-REN showed a rapid initial cost decrease, while the latter reached a lower final cost due to its memory units. All Youla policies outperformed their corresponding Feedback policies with the same model.

In Table~\ref{tab:computation} we compare the per-epoch computational cost of each method as a function of the number of neurons in the control policy. The proposed Youla-REN method required around two orders of magnitude less computation than the method of \cite{gu2021recurrent} which solves a semidefinite program at each gradient iteration. Computation time comparison with other models is omitted as their differences were negligable.

\begin{figure*}[t]
    \centering
    \includegraphics[width=0.96\textwidth]{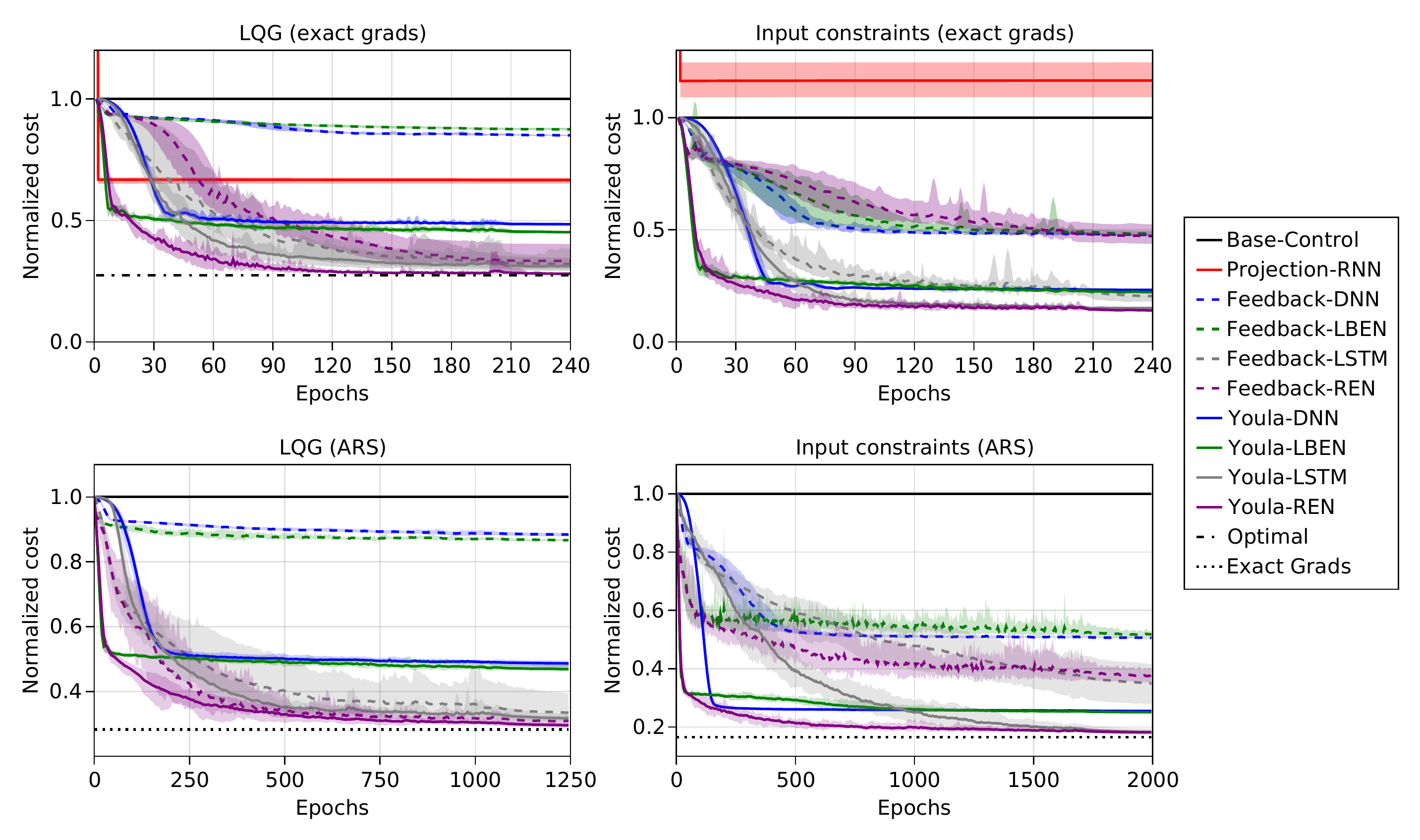}
    \caption{Normalized test cost vs epochs while learning the controllers with exact gradients and ARS. The mean cost computed over 10 random seeds is plotted, with colored bands showing the range. The base linear policy is an infinite-horizon LQG controller with random weight and covariance matrices. The optimal policy is the finite-horizon LQG controller with respect to the desired weight and covariance matrices. The final Youla-REN cost computed with exact gradients is also shown in the ARS plots for comparison (black dotted line). }\label{fig:result}
\end{figure*}
  
\begin{table}[!bt]
    \centering
    \caption{Computation time per epoch in seconds}
    \label{tab:computation}
    \begin{tabular}{c|ccccc}
        Number of neurons & 20 & 50 & 80 & 110 & 140 \\
        \hline
        Projection method \cite{gu2021recurrent} & 0.64 & 2.98 & 8.43 & 18.77 & 36.76\\
        Proposed method & 0.03 & 0.05 & 0.07 & 0.11 & 0.18
    \end{tabular}
\end{table}

\paragraph{Approximate Gradients}

Fig.~\ref{fig:result} also shows the results of learning controllers with ARS, where we can draw similar conclusions to the exact gradient experiments. The models achieved similar final costs to those trained with exact gradients in both tasks, even when limited to approximate gradients derived from policy evaluations at each iteration. Most notably, our Youla-REN controller still outperformed all other models.  

The ease of implementing approximate gradients with Youla-REN while still maintaining stability guarantees during learning is a direct consequence of our unconstrained parameterization. Each perturbation of the Youla-REN policy in ARS is intrinsically guaranteed to be stabilising, allowing ARS to take arbitrary perturbations in the policy parameters and remain within the set of stabilising controllers. This is a key feature of our approach, and provides a significant advantage over models like LSTM, which are not contracting for all parameters, or the projection method of \cite{gu2021recurrent}, which cannot be easily adapted to random search. It is worth noting that having stability guarantees during random search makes the Youla-REN architecture potentially suitable to learning on hardware in safety-critical systems, where trialling destabilizing policies during the training process must be avoided.

\bibliographystyle{IEEEtran}
\bibliography{refs}

\end{document}